\documentclass[10pt]{article}

\usepackage{epsf,amsfonts,amssymb,epsfig,amsmath,amsthm}
\usepackage{multirow}
\usepackage{dsfont}
\usepackage{bbold}

\usepackage{latexsym}
\usepackage{euscript}
\usepackage{mathrsfs} 
\usepackage{enumerate}
\usepackage{dsfont}

\usepackage{amssymb,amsmath,mathbbol,mathrsfs}
\usepackage[usenames,dvipsnames]{color}
\usepackage[linktocpage=true]{hyperref}
\usepackage{xcolor}
\usepackage{stmaryrd}
\usepackage{authblk}
\usepackage{framed}
\usepackage{empheq}
\usepackage{slashed}
\usepackage{hyphenat}
\usepackage{cite}

\usepackage{mdframed}

\renewcommand{\text}[1]{#1}

\newcommand{\be}{\begin{equation}}
\newcommand{\ee}{\end{equation}}
\newcommand{\ben}{\begin{displaymath}}
\newcommand{\een}{\end{displaymath}}
\newcommand{\bea}{\begin{eqnarray}}
\newcommand{\eea}{\end{eqnarray}}
\newcommand{\bean}{\begin{eqnarray*}}
\newcommand{\eean}{\end{eqnarray*}}
\newcommand{\ba}{\begin{array}}
\newcommand{\ea}{\end{array}}
\newcommand{\bi}{\begin{itemize}}
\newcommand{\ei}{\end{itemize}}

\newtheorem{lemma}{Lemma}
\newtheorem{theorem}{Theorem}
\newtheorem{corollary}{Corollary}

\newtheorem{definition}{Definition}

\newcommand{\EA}{{\mathfrak{A}}}
\newcommand{\EB}{{\mathfrak{B}}}

\newcommand{\frakh}{{\mathfrak{H}}}

\def\1f{f_1^{1/2}}
\def\2f{f_2^{1/2}}
\def\4f{f_4^{1/2}}

\def\rc{relativistic causality}
\def\RC{RC }

\def\caussep{causal severability}

\def\AUB{\hbox{$A \cup B$}}     
\def\linebreak{\hfil\break}

\def\Rbar{\bar{R}}

\def\Ehat{\widehat{E}}
\def\Fhat{\widehat{F}}

\begin{document}
\begin{titlepage}

\vfill



\begin{center}
  \baselineskip=16pt
  {\Large\bf  An intrinsic causality principle in histories-based quantum theory: a proposal}
 \vskip 1.5cm

Fay Dowker${}^{a,b}$ and Rafael D. Sorkin${}^{b,c,d}$\\
    \vskip .6cm
           \begin{small}
     \textit{ 
      ${}^{a}${Blackett Laboratory, Imperial College, Prince Consort Road, London, SW7 2AZ, UK}\\\vspace{5pt}
                ${}^{b}${Perimeter Institute, 31 Caroline Street North, Waterloo ON, N2L 2Y5, Canada}\\\vspace{5pt}
                ${}^{c}${{School of Theoretical Physics, Dublin Institute for Advanced Studies, 10 Burlington Road, Dublin 4, Ireland}}\\\vspace{5pt}
${}^{d}${{Department of Physics, Syracuse University, Syracuse, NY 13244-1130, U.S.A.}}\\\vspace{5pt}
  }
             \end{small}                       
\end{center}

\vskip 2cm
\begin{center}
\textbf{Abstract}
\end{center}
\begin{quote}
Relativistic causality (RC) is the principle that no cause can act
outside its future lightcone, but any attempt to formulate this
principle more precisely will depend on the foundational framework that
one adopts for quantum theory.
Adopting a histories-based (or ``path integral'') framework, we relate
RC to a condition we term ``Persistence of Zero'' (PoZ), according to
which an event $E$ of measure zero remains forbidden if one forms its
conjunction with any other event associated to a spacetime region that
is later than or spacelike to that of $E$.
We also relate PoZ to the Bell inequalities by showing that, in
combination with a second, more technical condition it leads to the
quantal counterpart of Fine's patching theorem in much the same way as Bell's
condition of Local Causality leads to Fine's original theorem.
We then argue that RC \textit{per se} has very little to say on the matter of
which correlations can occur in nature and which cannot.
From the point of view we arrive at,
histories-based quantum theories
are \textit{nonlocal in spacetime}, 
and \textit{fully in compliance with relativistic causality}.
\end{quote}

\vfill
\end{titlepage}

\tableofcontents

\newpage

\section{Introduction}\label{sec:intro}

Causal relationships are important in Quantum Field Theory
(QFT), in classical General Relativity (GR), and for Quantum
Foundations. Causality appeals to us as a basic scientific category, and
yet cause and effect lack a clear definition in quantum physics. There
is no consensus on whether  EPR-like correlations indicate that quantum
physics is nonlocal, or not relativistically causal, or both.

The meaning of cause is elusive, even classically.
In quantum mechanics and  QFT it is even harder to give it meaning. 
Causality concerns {\it\/events\/} in spacetime, whereas the field
operators in terms of which one usually 
formulates the causality conditions 
(as they are usually described)
of relativistic QFT
are not events.  
The only true spacetime event in an operator formulation of a quantum
theory is a ``measurement'' 
with its ensuing ``collapse of the
state-vector'',  but measurement and collapse rely on external observers
in spacetime. Measurement and collapse are needed in the canonical
operator-formulation but these concepts are not {intrinsic} to the quantum
system. In terms of the \textit{histories} of a system (trajectories in
quantum mechanics, field configurations in QFT), though, one
\textit{can} meaningfully speak of spacetime events without reference to
anything external to the quantum system. 
It is thus worthwhile to explore questions of causality and locality in
the context of a histories-based formulation of quantum
dynamics. 

Quantum Measure Theory (QMT) is such a framework and within it
one has the concept of an event in spacetime as a set of histories
\cite{Sorkin:1987cd, Sorkin:1994dt, Sorkin:1995nj, Sorkin:2006wq, Sorkin:2007uc, Sorkin:2010kg}.  
J.B. Hartle's Generalised Quantum Mechanics (GQM) 
\cite{Hartle:1989,Hartle:1991bb, Hartle:1992as, Hartle:2006nx} 
is a closely related histories-based framework for quantum foundations. The
two frameworks GQM and QMT are based on the same concepts of history,
of event (called a ``coarse-grained history'' in the GQM literature),
and
of decoherence functional, though they
diverge in their attitudes to
decoherence and probabilities, and in their interpretational
schemes. Since the technical results in this article do not depend on an
interpretational scheme, they are equally applicable to QMT and GQM.

This article will explore causality and locality in Quantum Measure
Theory, and in connection with EPR-type correlations. We will consider
only particle trajectories and field configurations on spacetime with no
appeal to external observers or agents. Our aim is to provide a dynamical axiom
that has some claim to be regarded as a principle of relativistic
causality.  
Thus, relativistic causality will be thought of as restricting a general class of
conceivable dynamics to a subclass that deserves to be described as causal.
We will assume, however, that spacetime is provided with a fixed, or
{``background''} relativistic causal structure; we will not
attack the important questions that arise when causal structure graduates
from background to dynamical, as it necessarily does whenever gravity is involved.

We begin in Section \ref{sec:qmt} by reviewing the basics of the histories-based, path integral inspired
framework of Quantum Measure Theory, including the concept of the event Hilbert space of a system. 
We  then specialise to the case of
quantum theory in a relativistic spacetime background in Section \ref{sec:est} and introduce the
restriction of  ``Persistence of Zero" on the dynamics of such a
quantum system.  Persistence of Zero (PoZ) is, as advertised in the title of this article, an {intrinsic} condition for a quantum theory, not reliant on concepts such as measurement by an external agent. We  motivate PoZ by showing that 
it implies, 
for each event, $E$, that there exists an  \textit{event-operator} $\hat{E}$ on the
event Hilbert space of the past (to be defined), such that $\hat{E}$
acting on the universal vector (also to be defined) produces the vector in the
event Hilbert space that corresponds to event $E$. PoZ implies that if
an event $E$ has measure zero then the event `$E$ and $F$' also has
measure zero whenever $F$ is an event that is nowhere to the past of
$E$.   We also show that the event-operators for spacelike events
commute.  

In Section \ref{sec:patch} we  further support the proposal that PoZ should be considered to be a causality condition. Specifically, we show that, in the case that the event Hilbert
space of the whole system equals the event Hilbert  space of the past---which condition we call  ``Lack of Novelty''---PoZ  plays the crucial role in the quantal
``patching theorem'' analogous to that played by the familiar factorizability condition in A. Fine's 
original classical version of the theorem \cite{Fine:1982}.

We follow this with two sections of more informal discussion. In Section \ref{sec:rescue} we argue that
relativistic causality, taken in isolation, is a rather weak condition
which in particular imposes no limitation on the spacelike correlations
that a localized cause can induce among events in its future. We illustrate this by reframing two well-known examples---the Popescu-Rohrlich box and the Greenberger-Horne-Zeilinger experiment---in the language of events, showing that there is nothing about the correlations in either of these cases that warrants the conclusion that relativistic causality is violated.  The only residual implication of relativistic causality, then, is that a localized cause should not induce
correlations between events that fail to be in its future, and this
seems to be what the PoZ condition is aiming to insure. However, as we discuss in section \ref{sec:separability},  PoZ
goes further by incorporating a limited amount of spacetime locality
that one might perhaps term as 
`\caussep', 
this being one way to
think of what the quantal patching theorem expresses.  

Finally, Appendix \ref{appendix} discusses a quantal analog of classical
factorizability that holds in relativistic QFT, and which could
conceivably take over the role of Lack of Novelty in an alternative
proof of the quantum patching theorem.


\section{Quantum measure theory}\label{sec:qmt}

QMT and GQM are path-integral inspired frameworks 
in which a quantum system is characterized by a triple, $(\Omega,
\mathfrak{A}, D)$, consisting of a set of histories $\Omega$, an event
algebra $\mathfrak{A}$ (a subalgebra of the power set of $\Omega$), and
a decoherence functional
$D(\cdot ,\cdot )$ on $\mathfrak{A}\times \mathfrak{A}$. 
In this section we review the basic concepts of history, event and decoherence functional and refer the reader to 
\cite{Sorkin:1987cd,Sorkin:1994dt,Sorkin:1995nj, Sorkin:2006wq, Sorkin:2007uc, Sorkin:2010kg} for more details on QMT, 
to  \cite{Hartle:1989,Hartle:1991bb, Hartle:1992as, Hartle:2006nx} for more details on GQM, 
and to \cite{Martin:2004xi, Dowker:2010ng,Dowker:2010qh,Sorkin:2012xx,Surya_2020} for more details on the event Hilbert space and the vector-measure.

\subsection{Event Algebra}

The kinematics of a quantum system in QMT is specified by 
the set $\Omega$ of \emph{histories} and one may have in mind that this is the 
set of histories over which the path integral is performed.
Each history $C$ in $\Omega$ is as complete a
description of the physical system as is in principle possible
in the theory. 
For example, in $n$-particle quantum mechanics, a history is a set of $n$
trajectories in spacetime and in a scalar field theory, a history is a real or complex
function on spacetime.

Any physical statement about (or property of) 
the system is a statement about (or property of) the history of the system. And, therefore, the statement/property corresponds to a subset of $\Omega$ in the obvious way. For example,
in the case of the non-relativistic particle, if $R$ is a region of
spacetime, the statement/property ``the particle passes through $R$'' corresponds to 
the set of all trajectories that pass through
$R$.  We adopt the terminology of stochastic
processes in which such subsets of $\Omega$ are referred to as \emph{events}. 

An \emph{event algebra} on a sample space $\Omega$ is a non-empty
collection, $\EA$, of subsets of $\Omega$ such that
\begin{enumerate}
\item  $\Omega\setminus E \in
\EA$ for all $E \in \EA$ (closure under complementation), 
\item $E \cup F \in
\EA$ for all  $E, F \in \EA $ (closure under finite union). 
\end{enumerate}
It follows from the definition that $\emptyset \in \EA$, $\Omega \in \EA$
and
$\EA$ is closed under finite intersections. An event algebra is an \textit{algebra of sets} by a standard definition,
and a \textit{Boolean algebra}. For events qua statements about the system, set operations correspond  to logical combinations of statements in the usual way: 
union is  ``inclusive or'', intersection is ``and'', complementation is ``not'' \textit{etc.} 

An event algebra $\EA$ is also an algebra in the sense of a vector space over a set of scalars, 
$\mathbb{Z}_2$, 
with intersection as multiplication and symmetric difference (or ``boolean sum'') as addition:
\begin{itemize}
\item[] $E F := E \cap F$, for all $E,  F \in \EA$;
\item[] $E + F := (E \setminus F) \cup (F \setminus
E) $ , for all $E, F \in \EA$.
\end{itemize}
In this algebra, the unit element, $\mathbb{1} \in \EA$, is the whole set of histories $\mathbb{1}: = \Omega$. The
zero element, $\mathbb{0}\in \EA$, is the empty set  $\mathbb{0}:=
\emptyset$. Note  that $E+F = E\cup F$ if and only if  $E$ and $F$ are
disjoint i.e. if and only if $EF = \mathbb{0}$. We will use this arithmetic
way of expressing set algebraic formulae whenever convenient, both for events and also for regions of spacetime.  

If an event algebra $\EA$ is also closed under countable unions then
$\EA$ is a $\sigma$-algebra but we do not impose this extra condition 
on the event algebra.
(See \cite{Dowker:2010qh,Sorkin:2012xx,Surya:2020cfm} for work on
extending the domain of the quantum measure from $\EA$ to a larger
subset of the $\sigma$-algebra generated by $\EA$.)

\subsection{Decoherence functional, quantum measure, vector-measure} \label{DecoherenceFunctional}
A  \emph{decoherence functional} on an event algebra $\EA$ is a map $D
: \EA \times \EA \to \mathbb{C}$ that encodes  both the initial conditions and the dynamics
of the quantum system and satisfies the conditions:
\begin{enumerate}
\item $D(E,F) =
D(F,E)^*$ for all $E, F \in \EA$ (\emph{Hermiticity});
\item $D(E,F + G) = D(E,
F) + D(E,G)$ for all $E, F, G \in \EA$  s.t. $FG= \mathbb{0}$ 
(\emph{Additivity});
\item $D(\Omega, \Omega)=1$ (\emph{Normalisation});
\item 
$D(E,E) \geq 0$ for all $E \in \EA$ (\emph{Weak Positivity}).
\end{enumerate}
See section 6 of \cite{Hartle:2006nx} for these axioms in the context of GQM.

\noindent
A \emph{quantum measure} on an event algebra $\EA$ is a map $\mu: \EA
\to \mathbb{R}$ such that,
\begin{enumerate}
 \item  $\mu(E) \geq 0$ for all $E \in \EA$
 (\emph{Positivity});
 \item $ \mu(E + F + G) - \mu(E + F) -
 \mu(F + G) - \mu(G + E) +\mu(E) +
 \mu(F) + \mu(G) = 0$, 
 for all $E, F, G \in \EA$ s.t. $EF=FG=GE= \mathbb{0}$ ({\emph{Quantum Sum Rule}});
 \item $\mu(\Omega)=1$ (\emph{Normalisation}).
\end{enumerate}

If $D : \EA \times \EA \to \mathbb{C}$ is a decoherence functional
then the map $\mu : \EA \to \mathbb{R}$ defined by
$\mu(E):=D(E,E)$ is a quantum measure. And, conversely, if $\mu$ is a quantum measure
on $\EA$  then there exists (a non-unique) decoherence functional $D$ such that 
$\mu(E)=D(E,E)$ \cite{Sorkin:1994dt}. 

The question of which is the primitive concept, quantum measure or
decoherence function, remains open. However, in this article we take the
decoherence functional to be the primitive concept because we will
assume all quantum systems satisfy the further axiom of strong
positivity which -- at our current level of understanding -- can only be
directly imposed on the decoherence functional:
\begin{definition}[Strong Positivity]
 A decoherence functional $D:  \EA \times \EA \to {\mathbb{C}}$
 is strongly positive if, for each finite set of events $\{E_1, \dots E_m\}$ 
 the corresponding $m \times m$ Hermitian matrix $M_{ij} := D(E_i, E_j)$ 
is positive semi-definite.  
\end{definition}
Strong positivity of the decoherence functional 
holds in established unitary quantum theories due to the form of the decoherence functional as a 
Double Path Integral (DPI) of Schwinger-Keldysh form, 
see for example equation (\ref{skform}) in Appendix \ref{appendix}.   
More generally, it has been shown that the set of strongly positive
quantum-measure systems is the unique set that is closed under
tensor-product composition and is ``full'' in the sense that if a system
can be composed with every element of the set, then that system is in the set  
\cite{strongpositivity2022}.

\subsection{Event Hilbert space and vector-measure}\label{sec:EHS} 
Henceforth in this article all quantum systems are assumed to have strongly positive decoherence functionals. Then, 
for each quantum system $(\Omega, \EA, D)$  a Hilbert space, $\frakh$, can be constructed as the completion of a quotient of the free vector space over the event algebra
\cite{Martin:2004xi,Dowker:2010ng}, 
such that for each event, $E\in\EA$, there is an \textit{event-vector}\footnote{In Generalized Quantum Mechanics, an event-vector is called a branch vector (e.g. \cite{Hartle:2006nx}).}
 $| E\rangle \in \frakh$ such that:
\begin{itemize}
\item $ \langle E \,|\, F \rangle = D(E,F), \ \forall E, F \in \EA$;
\item $ | E + F \rangle  = |E\rangle + | F \rangle, \ \forall$ $E, F \in \EA$ s.t. $EF = \mathbb{0}$;
\item $ \frakh$ is spanned by $\{ | E\rangle \}_{ E \in \EA} $. 
This spanning set is (very) over-complete: there is not a unique 
expansion of a vector in $ \frakh$ as a linear combination of event-vectors (see, for example, the previous point). 
\end{itemize}
We call this Hilbert space, $\frakh$, the \textit{event Hilbert space}. The 
corresponding map $| \cdot \rangle : \EA \rightarrow \frakh$ is called
the quantum vector-measure on the event algebra 
\cite{Martin:2004xi, Dowker:2010ng,Dowker:2010qh,Sorkin:2012xx,Surya_2020}.  

\begin{itemize}
\item An event has quantum measure zero if and only if its vector-measure is zero. 
\item If $\EB$ is a subalgebra of $\EA$ then the event Hilbert space of $\EB$ is a subspace of the event Hilbert space of $\EA$.
\item The vector-measure $|\Omega\rangle$ of the whole history space $\Omega$ is a unit vector in $\frakh$, and we
 call it \textit{the universal vector}.  If $\EB$ is a (unital) subalgebra of
 $\EA$ then the universal vector is in the event Hilbert space of
 $\EB$. The minimal subalgebra of $\EA$ is $\{\Omega, \mathbb{0}\}$, and
 its event Hilbert space is the 1-d vector space spanned by $|\Omega\rangle$.  
\item In a unitary quantum theory in which the decoherence functional
 is defined using the Schwinger-Keldysh double path integral with a
 pure initial state,
 the event Hilbert space coincides with the canonical Hilbert 
 space.\footnote{ \label{foot:EHS}
    Technically, they are naturally isomorphic.  At the level of
    mathematical theorems, 
    this remains to be proved in general. The existing theorems cover 
    the cases of nonrelativistic quantum mechanics and finite quantum
    systems  \cite{Dowker:2010ng}.}  
 In this case, the universal vector $|\Omega\rangle$  in the event Hilbert space is identified with
 the pure initial state in the canonical Hilbert space.
 If the decoherence functional is defined using a mixed initial state
 or ``density matrix''
 $\rho$ of rank $r$ then the event Hilbert space is a direct sum of $r$
 copies of the canonical Hilbert space.  
 The universal vector $|\Omega\rangle$ is then 
 the direct sum of the eigenvectors of $\rho$, 
 each normalised so that its norm squared is its eigenvalue \cite{Dowker:2010ng}. 
\item If $X$ is a family of events that constitute a partition of 
  $\Omega$\footnote{In Hartle's GQM this is called an exclusive, exhaustive set of coarse grained histories.}  
then 
\begin{align}
  \sum_{ E \in X } | E\rangle = |\sum_{ E \in X } E\rangle =  |\Omega\rangle\,.
\end{align} 
\end{itemize}

\section{Spacetime as an organising principle}\label{sec:est}

Our concern is relativistic quantum physics and we assume that the
histories of the system, the elements of $\Omega$, reside in a fixed
spacetime endowed with a causal structure that is mathematically a
partial order.  Any ``back reaction'' on the spacetime metric or causal
structure is thus being ignored.
This background may be a continuum such as a 4-d globally hyperbolic
spacetime or it may be a discrete partial order such as a causal set.

Let us call this spacetime $M$, and its causal order-relation $\preceq$.
The availability of this invariable substratum lets us organise events
in $\EA$ according to their location in $M$.
For each spacetime region $R\subseteq M$ there is a  subalgebra $\EA_R \subseteq \EA$,
where an  event $E$ is in $\EA_R$ iff  the property that defines
whether a history $\Gamma \in \Omega$ is in $E$ or not is a property of
$\Gamma$ in region $R$. In other words, if one can tell whether $\Gamma$
is in $E$ or not by examining the restriction of $\Gamma$ to $R$ then
$E$ is in $\EA_R$, otherwise it is not. We say `event $E$ {is in region}
$R$' to mean the same thing as `$E$ is an element of algebra $\EA_R$'. 
For each spacetime region $R$, the vector-measures of the events in $\EA_R$ span a subspace of the event Hilbert space: thus to each region
corresponds a sub-Hilbert space $\frakh_R$ of $\frakh$.
If two events
are located in mutually spacelike regions,
we will say that they are mutually spacelike events.

\subsection{Persistence of Zero}

Let us  consider a spacetime region $R$ and define the region $\Rbar$
to be the set of points that are not to the causal future of any point in $R$: 
\begin{align}
  \Rbar := M \setminus J^+(R) =  \{ y \in M :\, \nexists x \in R\  \textit{s.t.}\  x \preceq y\}  \,.
\end{align}
Note that  $\Rbar  R = \mathbb{0}$, since $ R \subseteq J^+(R)$.

If the dynamics is relativistically causal, then $\Rbar$ is the region
of spacetime that no event $E$ in $R$ can influence.  The question now
is whether this informal prohibition can be expressed (at least in part,
and subject to later revision when quantum-foundational questions are
better clarified) as a condition on the vector-measure of the system.

To that end, let us begin by asking whether it is possible to associate
to the event $E$ in $R$, not only a {\it\/vector\/} in $\frakh_R\subseteq\frakh$, 
but an {\it\/operator\/} on $\frakh$.   
We will denote this putative operator by $\Ehat$ and refer to it
as an \textit{event-operator}.%
\footnote{In the decoherent histories literature this would be called a
  \textit{class operator} see e.g. \cite{Hartle:2006nx,
  PhysRevD.73.024011}.  We know of nothing answering exactly to this
  description in algebraic quantum field theory, but one might
  hope to find it within the operator-algebra associated with region $R$.} 
In the first instance, however, it works better to limit the domain of 
$\Ehat$  to $\frakh_{\Rbar}\subseteq\frakh$. 

The Hilbert space $\frakh_{\Rbar}$ is spanned by the {event-vectors} 
$\{ |F\rangle:\, F \in \EA_{\Rbar}\}$. Let $E$ be an event in $R$ and
let us try to define a map $\Ehat$ from $\frakh_{\Rbar}$ to 
$\frakh$ by defining it on an arbitrary event-vector in $\frakh_{\Rbar}$ as:
\begin{align} \label{classop}
  \Ehat | F \rangle : = | E F \rangle \ \ \forall \ F \in \EA_{\Rbar}\,,
\end{align}
and then extending it by linearity to the whole of $\frakh_{\Rbar}$. $EF$ is the conjunction event `$E$ and $F$'.  
Given the causal relation between  $R$ and $\Rbar$,  although $E$ may be partly or wholly spacelike to $F$, $E$ cannot be to the causal past of $F$; it may therefore help to understand (\ref{classop}) to think of $EF$ as the event `$F$ \textit{and then} $E$'. 
Now (\ref{classop}) is not necessarily a consistent definition, because the expansion of a 
vector in $\frakh_{\Rbar}$ as a linear combination of event-vectors is not unique. 
This motivates the following definition: 
\begin{definition}[Persistence of Zero] \label{def:poz} 
  A quantum system satisfies Persistence of Zero (PoZ) if for every region $R$, every finite collection $C$ of events in 
  $\Rbar$, and every event $E$ in $R$
  we have 
\begin{align} \label{eq:pozero}
  \sum_{F \in C}  \psi_F | F \rangle = 0 \implies  
  \sum_{F \in C}  \psi_F | EF \rangle = 0 \,,
\end{align}
  where $\{\psi_F\}$ are complex coefficients. 
\end{definition}

It is characteristic of quantum theory that a subevent of an event
of measure zero can have nonzero measure due to interference: for
example in the iconic double slit experiment the measure of the event of
the particle arriving at a dark fringe is zero but the measure of the
event of the particle passing through the left slit before arriving at a
dark fringe is nonzero. But  PoZ implies that if an event $F$ has
measure zero then any subevent of the form $F E$, where $E$ is an event
that is \textit{future and/or spacelike} to $F$,  
also has measure zero. 
So we can already see that PoZ is some sort of causality condition. 
Moreover, it is just what the definition of $\Ehat$ needs for consistency: 
\begin{lemma}\label{PoZ} If a quantum system satisfies PoZ then (\ref{classop}) defines, 
   for each event $E$ in a region $R$ such that $\Rbar$ is nonempty, 
   a linear map, the event-map  $\Ehat: \frakh_{\Rbar} \rightarrow \frakh$. 
\end{lemma}
\begin{proof}
The event-vectors in $\frakh_{\Rbar}$ span $\frakh_{\Rbar}$ and 
condition PoZ is exactly the condition that $\Ehat$ is well defined when extended by linearity to all 
linear combinations of event-vectors.\footnote{This proof is rigorous under the assumption that the
   event algebra $\EA$  is finite (i.e. that $\Omega$ comprises only a
   finite number of histories).  In the contrary case, the sums in
   (\ref{eq:pozero}) are only guaranteed to fill out a dense subspace of the event Hilbert space.  We
   will ignore all such technicalities in this paper, recalling in
   this connection that similar technicalities affect the
   definition of the path-integral itself.} 
\end{proof}

\begin{corollary} \label{corPoZ}
 If a quantum system satisfies PoZ then (\ref{classop}) defines, for each event $E$ in a region $R$ such that $\Rbar$ is nonempty, 
 and each subregion $Q \subseteq \Rbar$, 
the event-map,  
 $ \Ehat:  \frakh_{Q} \rightarrow \frakh_{Q + R}$. 
\end{corollary}

\begin{corollary} 
If a quantum system satisfies PoZ then event-map $\Ehat$ acting on the universal vector equals the event-vector for $E$:
\begin{align} 
  \Ehat|\Omega\rangle = | E \Omega\rangle = | E\rangle\,.
\end{align}
\end{corollary}

\begin{corollary} \label{linear}
  If a quantum system satisfies PoZ and $E$ and $F$ are disjoint events in a region $R$ then $\widehat{E+F} = \Ehat + \Fhat$. 
\end{corollary}

\begin{corollary} \label{sumto}
  If a quantum system satisfies PoZ and $\{E_1, E_2\dots E_n\}$ is a collection of events in a region $R$ that is a
  partition of $\Omega$, then the  event-maps 
  $\{\Ehat_1, \Ehat_2, \dots \Ehat_n\}$ 
  sum to the identity map from $\frakh_{\Rbar}$  
  to  $\frakh_{\Rbar}$ considered as a subspace of $\frakh$.   
\end{corollary}

\subsection{Spacetime arrangement in the case of interest}

In this article we are interested in the following physical setup in
spacetime.  As shown in figure \ref{frog},
$Z$ is a region that contains its own causal past: $Z$ is a
`past set'.  Let  regions $A$ and $B$ both  lie in the future domain of
dependence of $Z$ and be disjoint from $Z$ (recall that the future domain of
dependence of $Z$ includes $Z$) and be spacelike to each other.  The
union of regions $Z$, $A$ and $B$ is also a past set. The heuristic
justification for this arrangement is that,
in a relativistically causal theory,
any cause of a
correlation between events in $A$ and $B$ must be in $Z$.  
In particular any ``preparation event'' in an experiment of EPR type
will automatically be contained within the region $Z$.

\begin{figure}[h!]
\centering
{\includegraphics[scale=0.25]{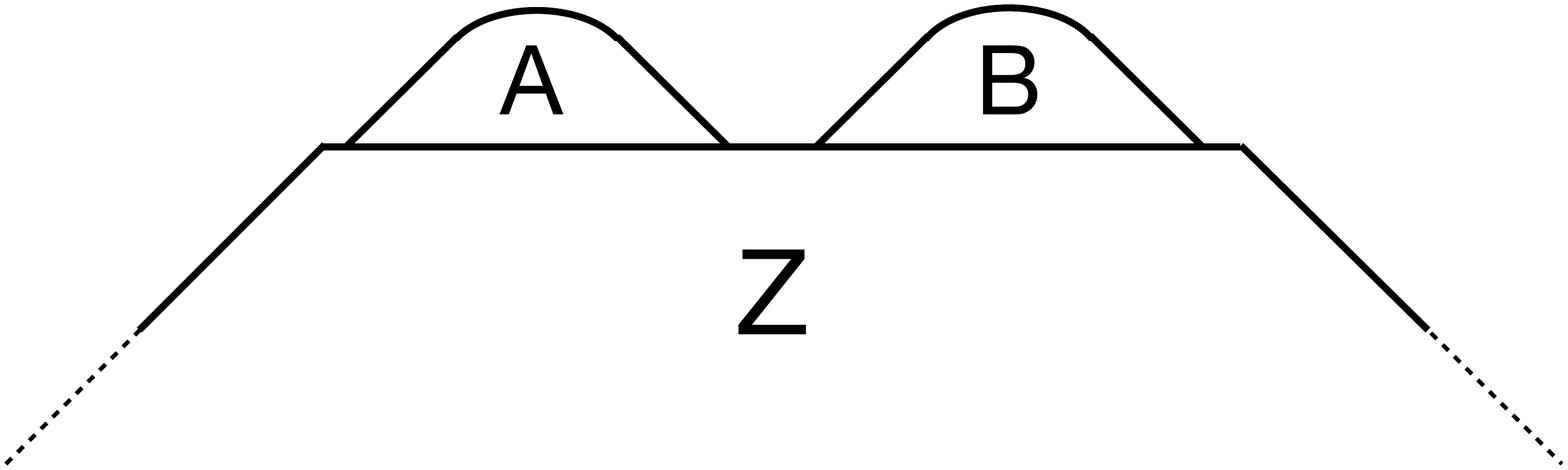}}
\caption{\label{frog} 
  $Z$ is a past set (i.e. contains its own causal past). 
  $A$ and $B$ are in the future domain of dependence of $Z$, do not intersect 
  $Z$ and are spacelike to each other. 
  The union of $Z$, $A$ and $B$ is also a past set.}
\end{figure}

In unitary quantum field theory in a globally hyperbolic spacetime, the
event Hilbert space of the future domain of dependence of $Z$ equals the
event Hilbert space of $Z$. For example, if $Z$ is the past of a Cauchy 
surface then the
event Hilbert space of $Z$ and the event Hilbert space of the future domain of dependence of $Z$ both equal the canonical Hilbert space of the whole system.\footnote%
{As mentioned previously in footnote \ref{foot:EHS} the existing theorems cover the cases of nonrelativistic quantum mechanics and finite quantum systems. The formal extension to QFT would follow the same proof-structure: showing that the physically induced map from the event Hilbert space to the canonical Hilbert space--which is injective because it preserves the inner product---is surjective  \cite{Dowker:2010ng}.}

This can be thought of as a condition of `lack of novelty': the past
region $Z$ is rich enough in events 
that further events anywhere in the
future domain of dependence of $Z$ do not add anything new to the
physics as
encoded in the event Hilbert space. 
We formalise this
condition for quantum measure theories in general:
\begin{definition}[Lack of Novelty] \label{def:lon} 
A quantum system satisfies Lack of Novelty (LoN) if,  for every region of spacetime $Z$ that contains its own past, 
the event Hilbert space of the future domain of dependence of $Z$ equals the event Hilbert space of $Z$.
\end{definition}

In the case of the regions shown in figure \ref{frog} LoN implies that 
Hilbert spaces $\frakh_{Z+A}$, $\frakh_{Z+B}$ and  $\frakh_{Z+A+B}$ all equal $\frakh_{Z}$, where 
subscript $Z+A$ refers to the union of regions $Z$ and $A$, and so on. 
So, for example, if $E_A$ is an event in $A$ then 
there exist complex coefficients $\{\psi_F\}$ such that 
\begin{align}
 |E_A \rangle = \sum_{F \in \EA_Z} \psi_{F} | F\rangle\,.
\end{align}

\begin{lemma}\label{commute}
Let the quantum system satisfy PoZ (definition \ref{def:poz}) and LoN (definition \ref{def:lon}).  Let regions $Z$, $A$ and $B$ be as described in figure \ref{frog}. 
If $E_A$ is an event in $A$ and 
$E_B$ is an event in $B$, then their respective event-maps, $\Ehat_A$ and $\Ehat_B$, are operators on 
$\frakh_Z$,  and they commute. 
\end{lemma}
\begin{proof} 
Corollary  \ref{corPoZ} gives us the event-map
$\Ehat_A: \frakh_Z \rightarrow  \frakh_{Z+A}$. 
Since $\frakh_{Z+A} = \frakh_Z$,  $\Ehat_A$ is an operator on $\frakh_{Z}$. 
Similarly for $\Ehat_B$. Henceforth we refer to such operators as event-operators. 

\noindent
Then, for each $E_Z \in \EA_Z$, 
\begin{align}
  \Ehat_B \Ehat_A  |E_Z\rangle & =  \Ehat_B |E_A E_Z  \rangle \\
  &= |E_B E_A E_Z   \rangle\\
  &  =\Ehat_A  \Ehat_B |E_Z \rangle\,.
\end{align}
The first equality follows by applying Corollary \ref{corPoZ} to region $A$ and the 
second equality  follows by applying Corollary \ref{corPoZ} to region $B$. 
We also used the fact that multiplication commutes in the Boolean algebra so $E_B E_A E_Z =  E_A E_B E_Z$. 
Since the event-vectors $| E_Z\rangle$  span $\frakh_Z$ the result follows. \\
\end{proof}


\section{Patching theorems}\label{sec:patch}

In this section we prove 
the quantum analog of Fine's patching theorem
(Proposition 3 of \cite{Fine:1982}), taking PoZ and LoN as our inputs.

Let us consider, for definiteness and following Fine, the Clauser-Horne-Shimony-Holt (CHSH)
scenario \cite{Clauser:1969} with two spin-half particles and two
``local experiments'', one in each of two spacelike wings.  Each local
experiment takes the form of a Stern Gerlach analyzer that admits of two possible
orientations or ``settings'', and through which the particle emerges in
one or the other of two exit beams, the upper beam or the lower
beam, where it registers in a detector.
This scenario can be generalised to any number of spacelike separated
regions with any number of settings per analyzer and any number of beams
per setting (the particles might have different spins, or there might be
sequences of concatenated beam-splitters), and our results generalise
mutatis mutandis.

Let us call the two spacelike separated regions where the local
experiments take place $A$ and $B$, and let us call the possible
settings, $a$ or $a'$ in $A$ and $b$ or $b'$ in $B$. We refer to $a$,
$a'$, $b$ and $b'$ as local 
settings. There are
then four possible global settings: $ab$, $ab'$, $a'b$
and $a'b'$.

The two spin-half particles are prepared and sent to the analyzers, one
particle to $A$ and one particle to $B$.   Let $Z$ be a region of
spacetime such that $Z$ is a past set and such that $A$ and $B$ lie in
the future domain of dependence of $Z$ without intersecting $Z$, and the
union of $A$, $B$ and $Z$ is also a past set, as shown in figure
\ref{frog}.  

For each local setting, $a$ say for the analyzer in $A$, there are two
possible beams in which the particle can be detected and two beam-events, corresponding to the particle being detected in the upper beam (u) and the lower beam (d) respectively. There are
four sets of experimental probabilities, one for each global setting:  
\begin{align} \label{expprobs}
\mathbb{P}^{ab}(ij)\,, 
\mathbb{P}^{ab'}(ij')\,, 
\mathbb{P}^{a'b}(i'j)\,, \  \textrm{and}\ 
\mathbb{P}^{a'b'}(i'j')\,, 
\end{align}
where $i, i', j, j'$  labels the event of the 
particle being detected in the upper beam or lower beam  (each label taking values $u$ or $d$ respectively) for setting 
$a, a', b, b'$ respectively.  We use the shortened term ``beam-event" to
signify the detection of a particle in a particular beam, and as above,
we sometimes use a simplified notation in which the labels $i,i', j, j'$ 
stand for the beam-events itself.  

We now postulate that the four probability distributions are compatible with each
other in the sense that $\mathbb{P}^{ab}(j) = \mathbb{P}^{a'b}(j)$, etc.
These conditions (the so-called no-signalling conditions) signify that
the experimental probabilities of the events in $B$ do not depend on the
setting in $A$ and vice versa.

The patching theorems are about different measure theories in the same region of spacetime and we will need the following concepts and definitions. For a measure theory in a spacetime $M$ with history-space $\Omega$, if  $R\subset M$ is a region of spacetime, then a history in $R$ is an element of $\Omega$ restricted to $R$. We use the notation $\Omega_R = \{ \Gamma|_R :\, \Gamma \in \Omega\}$ for the space of histories restricted to $R$. 
\begin{definition} [History-event in a region]  Let  $R\subseteq M$ be a region of spacetime,   a history-event  in $R$ is a cylinder set,
\begin{align}
E_{\gamma_R} := \{ \Gamma \in \Omega : \, \Gamma |_R = \gamma_R\}\,,
\end{align}
where $\gamma_R$ is a history in $R$.
\end{definition} 

\begin{definition}[Agreement of measure theories in a region] \label{def:agree} Two measure theories, $\{\Omega^1, \EA^1, D^1\} $ and  $\{\Omega^2, \EA^2, D^2\} $, in $M$ \textit{agree in region} $R$ if:
\begin{itemize}
\item[i.] The two history-spaces restricted to $R$ are equal, 
\begin{align}\Omega^1_R = \Omega^2_R\,,
\end{align} 
in which case there is an obvious physical isomorphism between  $\EA^1_R$ and $\EA^2_R$, the event algebras for the two theories restricted to $R$. 
\item[ii.] The decoherence functionals $D^1$ and $D^2$, restricted to $\EA^1_R$ and $\EA^2_R$ respectively---denoted $D^1_R$ and $D^2_R$---are equal via the isomorphism. i.e. if the physical isomorphism is 
\begin{align}\phi: \EA^1_R\rightarrow \EA^2_R
\end{align} 
then 
\begin{align}
D^1(E, \bar{E}) = D^2(\phi(E), \phi(\bar{E}))\,, \  \forall E, \bar{E} \in \EA^1_R\,.
\end{align}

\end{itemize} 
\end{definition}

When two measure theories $1$ and $2$ agree in a region $R$, then we can and will henceforth identify the algebras $\EA^1_R$ and $\EA^2_R$ via the physical isomorphism. And then the decoherence functionals are equal in $R$: $D^1_R = D^2_R$,
and further, the sub-event Hilbert spaces in $R$ are equal: $\frakh^1_R = \frakh^2_R$.  Strictly, since $\frakh^1_R$ and $\frakh^2_R$ are subspaces of different event Hilbert spaces in different theories, they are different spaces but we can and will identify them. 

\subsection{Fine's classical patching theorem}

In the framework of QMT, a \textit{classical theory} is a quantum measure theory $(\Omega, \EA, D)$ that satisfies the additional condition
\begin{align}
  \label{classical}
  D(E, F) = D(E F, E F) \ \ \forall \ E, F \in \EA\,.
\end{align}
In a classical theory, all the information in the decoherence functional is encoded in the measure,
$ \mu (E) = D(E, E)$, which measure satisfies the Kolmogorov sum rule and is referred to as a classical measure or equivalently as a probability measure. 
A classical measure theory $(\Omega, \EA, \mu)$ is a level 1 theory in the hierarchy of measure theories delineated in \cite{ Sorkin:1994dt}.

\begin{definition}[Factorizability] \label{def:factor}
  A classical measure theory  $(\Omega, \EA, \mu)$ on a spacetime $Z+A+B$ where regions $Z$, $A$ and $B$ are as
  in figure \ref{frog},   is factorizable if   
  \begin{align}\label{factorizable}
      \mu(E_A E_B E_{\gamma_Z}) \; \mu(E_{\gamma_Z}) = \mu(E_A E_{\gamma_Z}) \; \mu(E_B E_{\gamma_Z})
  \end{align}
  for all $E_A \in \EA_A$ and $E_B \in \EA_B$ and all history-events $E_{\gamma_Z}$ in $\EA_Z$. 
\end{definition}

\noindent
If $\mu(E_{\gamma_Z})$ is nonzero, and if we divide through by its
square, then this becomes the statement (sometimes called `screening
off') that 
the joint probability of $E_A$ and $E_B$ factorizes
when conditioned on a  history-event in $Z$.

\begin{definition} [Factorizable model]  \label{def:facmod}
A \textit{factorizable model} for the probabilities  (\ref{expprobs}) in the 
CHSH scenario is a set of 
four factorizable classical measure theories in $Z+A+B$, 
$(\Omega^{\alpha\beta}, \EA^{\alpha\beta}, \mu^{\alpha\beta})$ labelled by the four global settings $\alpha\beta = ab,ab',a'b, a'b'$, 
 with the following properties. 
\begin{itemize} 
 \item[i.] For each local setting, $a,a',b$ or $b'$, the two theories that share that setting agree (as per definition \ref{def:agree}) in the relevant spacetime region:  theories $ab$ and $ab'$ agree in $Z+A$, theories $a'b$ and $a'b'$ agree in $Z+A$, theories $ab$ and $a'b$ agree in $Z+B$, and theories $ab'$ and $a'b'$ agree in $Z+B$.
 This implies that all four theories agree in $Z$. 
 \item[ii.] For each local setting, the particle is detected either
  in the upper beam or in the lower beam. For example, for local setting $a$ let  $E^a_{u}$ and $E^a_{d}$ be the beam-events where we have implemented our declared intention to identify events in regions where theories agree. Then 
 $E^a_{u}\cup E^a_{d} =\Omega^{ab}$ in measure theory $ab$ and  $E^a_{u}\cup E^a_{d}= \Omega^{ab'}$ in measure theory $ab'$.\footnote{This doesn't imply that $\Omega^{ab} = \Omega^{ab'}$. Strictly, the event $E^a_{u}$ in measure theory $ab$ and the event  $E^a_{u}$
  in measure theory $ab'$ are events in two different theories and the histories in $\Omega^{ab}$ and $ \Omega^{ab'}$ differ in $B$.}
    Similarly for each of the other local settings: e.g. for setting $a'$ we have $E^{a'}_{u}\cup E^{a'}_{d} = \Omega^{a'b}$ in theory $a'b$ and  $E^{a'}_{u}\cup E^{a'}_{d} = \Omega^{a'b'}$  in theory $a'b'$,  and so on. 
\item[iii.] Each measure $\mu^{\alpha\beta}$ has the corresponding experimental 
  probabilities $\mathbb{P}^{\alpha\beta}$ (\ref{expprobs}) as marginals:
  \begin{align}
    \sum_{k} \mu^{ab}(E^a_iE^b_jE^Z_{\gamma_k}) &= \mathbb{P}^{ab}(ij)\,,\nonumber\\
        \sum_{k} \mu^{ab'}(E^a_iE^{b'}_{j'}E^Z_{\gamma_k}) &= \mathbb{P}^{ab'}(ij')\,,\nonumber\\
            \sum_{k} \mu^{a'b}(E^{a'}_{i'}E^b_jE^Z_{\gamma_k}) &= \mathbb{P}^{a'b}(i'j)\,,\nonumber\\
                \sum_{k} \mu^{a'b'}(E^{a'}_{i'}E^{b'}_{j'}E^Z_{\gamma_k}) &= \mathbb{P}^{a'b'}(i'j')\,,\label{eq:probmarg}
  \end{align}
where $k$ labels the histories in $Z$ and $i,j, i',j' = u,d$. 
\end{itemize}
\end{definition}

\vskip.5cm

\noindent
Now we can state 
\begin{theorem}[Fine's patching theorem \cite{Fine:1982}]  If there
 exists a factorizable model (Definition  \ref{def:facmod}) for the probabilities (\ref{expprobs}) then  there exists 
  a joint probability measure on all the
  beam-events and all events in $Z$ that has the four 
  factorizable measures, $\{\mu^{\alpha\beta}\}$ 
  as marginals and that therefore (by (\ref{eq:probmarg})) has the probability measures (\ref{expprobs})  as further marginals. 
\end{theorem}
\begin{proof} We will use $i,i',j,j'$ and $k$ as shorthand for the events $E^{a}_{i}, E^{a'}_{i'}, E^{b}_{j}, E^{b'}_{j'}$ and $E^Z_{\gamma_k}$. \\

First note that due to the agreement of the four  measures we have $\mu^{ab}(ik) = \mu^{ab'}(ik)$ and so we can define
$\mu^{a}(ik) : = \mu^{ab}(ik)= \mu^{ab'}(ik)$. And similarly for each of the other 3 local settings, $a'$, $b$ and $b'$. 
Similarly we can define $\mu(k) : = \mu^{ab}(k) = \mu^{a'b}(k)=\mu^{ab'}(k)=\mu^{a'b'}(k)$.  

Now, define a joint probability measure on all the beam-events and the $Z$-history-event $k$:
\begin{align}\label{jointp}
  \mu^{patch}(i i' j j' k): = \mu^{a}(ik)\mu^{a'}(i'k)\mu^{b}(jk)\mu^{b'}(j'k) / \left(\mu(k)\right)^{3}\,.
\end{align}
(\ref{jointp}) is well defined 
since if $\mu(k) = 0$ then all the probabilities in the numerator vanish too
(in which case we  define $\mu^{patch}(i i' j j 'k)$ to equal zero).
The factorizability of the four measures $\{\mu^{\alpha\beta}\}$ implies that $\mu^{patch}$ has
each of them as marginals. For example,
\begin{align}
\sum_{i'j'}  \mu^{patch}(i i' j j' k)& = \sum_{i'j'} \mu^{a}(ik)\mu^{a'}(i'k)\mu^{b}(jk)\mu^{b'}(j'k) \left(\mu(k)\right)^{-3}\\
& =  \mu^{a}(ik)\mu^{a'}(k)\mu^{b}(jk)\mu^{b'}(k) \left(\mu(k)\right)^{-3}\\
& =  \mu^{a}(ik)\mu^{b}(jk)\left(\mu(k)\right)^{-1}\\
& =  \mu^{ab}(ijk) \,.
\end{align}
The calculation for the other three global settings is similar. \\
Summing further over the history-events $k$ in $Z$ then gives, by (\ref{eq:probmarg}),  the probabilities (\ref{expprobs}) as marginals.
\end{proof}

By summing (\ref{jointp}) over the history-events in $Z$, one obtains a patched probability measure on the beam-events alone:
\begin{align}\label{eq:jointpexp}
  \mu^{patch}(i i' j j') = \sum_k  \mu^{patch}(i i' j j' k)\,,
\end{align}
and Fine further shows that the existence of such a probability measure implies the CHSH-Bell inequalities \cite{Fine:1982}. \\

We specialised to the CHSH scenario for concreteness and so that we didn't have to introduce complicated notation for arbitrary scenarios as for example in \cite{Dowker:2013tva}. The extension of Fine's patching theorem to general scenarios is straightforward. Finding the corresponding generalised Bell inequalities, on the other hand, is a hard problem.

\subsection{A quantum patching theorem}

Factorizability (\ref{factorizable}) has a direct analog in QMT \cite{Craig:2006ny}:
\begin{align} \label{qfactor} 
    D(E_A E_B  E_{\gamma_Z},\bar{E}_A \bar{E}_B E_{\bar{\gamma}_Z}) \; D(E_{\gamma_Z}, 
E_{\bar{\gamma}_Z}) 
  = D(E_A  E_{\gamma_Z}, \bar{E}_A  {E_{\bar{\gamma}_Z}})\;  D(E_B  E_{\gamma_Z}, \bar{E}_B  
  E_{\bar{\gamma}_Z}) \,,
\end{align}
for all $E_A, \bar{E}_A \in \EA_A$, $E_B, \bar{E}_B \in \EA_B$, and all history-events $E_{\gamma_Z}, E_{\bar{\gamma}_Z} \in \EA_Z$.

This looks just like factorizability in the classical
case except ``doubled" with each relevant event in $A$ or $B$ or $Z$
replaced by a pair of events in $A$ or $B$ or $Z$, respectively ($E_A$
replaced by $(E_A, \bar{E}_A)$, for example). And, promisingly,  this
``quantum factorizability'' condition is satisfied formally by
relativistic QFT, 
as we show in Appendix \ref{appendix}. 
However,  the
patched joint decoherence functional for the CHSH scenario obtained by
using the analog of the formula of (\ref{jointp})   fails  to exist in
general because the numerator does not necessarily vanish when the
denominator vanishes.   
Nor is it necessarily strongly positive even if it is defined \cite{Craig:2006ny}. 

In \cite{Craig:2006ny} it was shown that in ordinary quantum mechanics,
the existence of projection operators for the beam-events in each of the local
settings allows the construction of a patched decoherence functional.  
Since event-operators can take the place of projection operators in the
derivation, and since PoZ enables the construction of event-operators,
we look now to PoZ instead of factorizability as the basis of patching.

In the quantum CHSH scenario most generally, instead of 4 probability measures on the beam-events there are 4 decoherence functionals, one for each global setting: 
\begin{align}
D^{ab}(ij , {\bar{i}} {\bar{j} })\,, 
D^{a'b}(i'j, \bar{i}'\bar{j})\,, 
D^{ab'}(ij', \bar{i}\bar{j}')\,, \ \textrm{and} \ 
D^{a'b'}(i'j', \bar{i}'\bar{j}')\,.   \label{fourdcfs}
\end{align}
When the $i,i',j,j'$ refer to  beam-events that are detection events that behave classically, these decoherence functionals will be diagonal, and the diagonal elements will be the experimental probabilities (\ref{expprobs}). For the quantum patching theorem below, however, this diagonal property is not needed and so we will consider (\ref{fourdcfs})
 merely as 4 decoherence functionals that are compatible on common events (``no-signalling''),  i.e. 
\begin{align}
\sum_{j \bar{j}} D^{ab}(ij , {\bar{i}} {\bar{j} }) 
= \sum_{j' \bar{j}'}D^{ab'}(ij', \bar{i}\bar{j}')\,,
\end{align}
and all similar conditions. 

We now define a PoZ model for decoherence functionals by analogy with factorizable model for probability measures (definition \ref{def:facmod}): 
\begin{definition}[PoZ model] \label{def:pozmodel} A \textit{PoZ model}
for the decoherence functionals  (\ref{fourdcfs}) is a set of four PoZ quantum measure theories in spacetime $Z+A+B$,
 $(\Omega^{\alpha\beta }, \EA^{\alpha\beta }, D^{\alpha\beta Z})$ labelled by the four global settings $\alpha\beta = ab,ab',a'b, a'b'$, 
 with the following properties.
 \begin{itemize}
  \item[i.] For each local setting, $a,a',b$ or $b'$, the two theories that share that setting agree (as per definition \ref{def:agree}) in the relevant spacetime region, i.e.  theories $ab$ and $ab'$ agree in $Z+A$, theories $a'b$ and $a'b'$ agree in $Z+A$, theories $ab$ and $a'b$ agree in $Z+B$, and theories $ab'$ and $a'b'$ agree in $Z+B$.
 This implies that all four theories agree in $Z$. 
\item[ii.]  For each local setting, the particle is detected either
  in the upper beam or in the lower beam. For example, for local setting $a$ 
 $E^a_{u}\cup E^a_{d} =\Omega^{ab}$ in measure theory $ab$ and  $E^a_{u}\cup E^a_{d}= \Omega^{ab'}$ in measure theory $ab'$. 
    Similarly for each of the other local settings: e.g. for setting $a'$ we have $E^{a'}_{u}\cup E^{a'}_{d} = \Omega^{a'b}$ for in theory $a'b$ and  $E^{a'}_{u}\cup E^{a'}_{d} = \Omega^{a'b'}$  in theory $a'b'$,  and so on.  
\item[iii.] Each decoherence functional $D^{\alpha\beta Z}$ has the corresponding decoherence functional $D^{\alpha\beta}$ (\ref{fourdcfs}) as marginals. For example:
  \begin{align}
    \sum_{k \, \bar{k} } D^{ab Z}(ijk, \bar{i}\bar{j}\bar{k})
    ) = D^{ab}(ij, \, \bar{i}\bar{j})\,,\label{eq:dcfmarg}
  \end{align}
 where $k$  is shorthand for the history-events $\gamma_k$ in $Z$ labelled by $k$. And similarly for $a'b$, $ab'$ and $a'b'$.
\end{itemize}

\end{definition}

PoZ implies that there are four Hilbert spaces, $\frakh^{abZ}$, $\frakh^{ab'Z}$, $\frakh^{a'bZ}$
 and $\frakh^{a'b'Z}$, one for each global setting. Since all the four PoZ theories agree in $Z$, these four Hilbert spaces each contain as a subspace the Hilbert space $\frakh_Z$ that is spanned by the event-vectors $\{|E_Z  \rangle :\, E_Z \in \EA_Z\}$. The additional LoN assumption then implies that $\frakh^{abZ}=\frakh^{ab'Z}=\frakh^{a'bZ} = \frakh^{a'b'Z} =\frakh_Z$. Then by Lemma \ref{commute}, for each beam-event for each local setting there exists an operator on  $\frakh_Z$,  so there are 8 event-operators on $\frakh_Z$---$\{\hat{E}_i^a, \hat{E}_{i'}^{a'}, \hat{E}_j^b, \hat{E}_{j'}^{b'}\}$---and the event-operators for beam-events in $A$ commute with the event-operators for 
 beam-events in $B$. 
 \begin{lemma}\label{beamops}
If there is a PoZ model for the four decoherence functionals (\ref{fourdcfs}) and if LoN holds for the four measure theories in the model, then for each local setting, $a,a',b$ or $b'$, the two 
event-operators on $ \frakh_Z$ corresponding to the up and down beam-events for that setting sum to the identity operator. For example for setting $a$,  $\hat{E}^a_{u} 
+ \hat{E}^a_{d} = \mathbb{1}$.
\end{lemma}
\begin{proof}  Consider an event-vector,  $| E_Z \rangle$, in $\frakh_Z$. Choose local setting $a$ and consider the two corresponding beam-events $ {E}^a_{u}$ and ${E}^a_{d}$. By condition (ii) in the definition of a PoZ model, $ {E}^a_{u} + {E}^a_{d} = \Omega$ in each of the  measure theories $ab$ and $ab'$ that 
$a$ is a part of.  

Then, 
\begin{align}
(\hat{E}^a_{u} 
+ \hat{E}^a_{d}) | E_Z \rangle 
& =  | ( {E}^a_{u} + {E}^a_{d}) E_Z \rangle\\
 &= | \Omega E_Z \rangle\\
& = |E_Z \rangle\,.
\end{align} 
The event-vectors span the Hilbert space $ \frakh_Z$ and so $\hat{E}^a_{u} 
+ \hat{E}^a_{d} = \mathbb{1}$ on $ \frakh_Z$. There is a similar proof for each of the other local settings, $a'$, $b$ and $b'$. 
\end{proof}

 \begin{theorem}[Quantum patching theorem]  If there exists a PoZ model (Definition \ref{def:pozmodel}) for 4 decoherence functionals (\ref{fourdcfs}) and  
each of the 4 measure theories $D^{\alpha\beta Z}$ in the PoZ model also satisfies LoN for $\alpha\beta = ab, ab', a'b, a'b'$ 
  then there exists, mathematically,  a joint decoherence functional on all the beam-events and all events in $Z$ that has the 4 PoZ model decoherence functionals, $D^{\alpha\beta Z}$, as  marginals. 
\end{theorem}
\begin{proof} 

By lemma \ref{commute}, there are 8 event-operators on Hilbert space $\frakh_Z$, $ \hat{E}^a_i$, $ \hat{E}^{a'}_{i'} $, $\hat{E}^b_j $ and $\hat{E}^{b'}_{j'} $ 
for $i,i',j,j' = u,d$ such that  the operators in $A$ commute with the operators in $B$.  

 From these we can define, for each history-event-vector $|{E}^Z_k  \rangle$ in $\frakh_Z$ labelled by $k$, 
 the 16 vectors,
\begin{align}\label{jointvectors}
| ii'jj'k \rangle :=
\hat{E}^a_i \hat{E}^{a'}_{i'} \hat{E}^b_j \hat{E}^{b'}_{j'}    |{E}^Z_k  \rangle\,.
\end{align}
Using  these vectors we define a patched joint decoherence functional,
\begin{align}\label{jointdcf}
D^{patch}( ii'jj'k\,,\, \bar{i} {\bar{i}}'\bar{j}{\bar{j}}'\bar{k} ) &:=
\langle  ii'jj'k \, |\,  \bar{i} {\bar{i}}'\bar{j}{\bar{j}}'\bar{k} \rangle \,.
\end{align}
$D^{patch}$ depends on the  choice of ordering of the event-operators in the string in (\ref{jointvectors}). However, any ordering will work in the following calculation because the event-operators for events in $A$ commute  with the event-operators for events in $B$.\\

We now show that $D^{patch}$ has $D^{\alpha\beta Z}$, for each $\alpha\beta$, as marginals. For example, for $\alpha\beta = a b$,
\begin{align}\label{marginalsofpatch}
\sum_{ i' j'} | ii'jj'k \rangle &= \sum_{ i' j'} 
\hat{E}^a_i \hat{E}^{a'}_{i'} \hat{E}^b_j \hat{E}^{b'}_{j'}    |{E}^Z_k  \rangle\\
& = \hat{E}^{a}_i ( \hat{E}^{a'}_{u} + \hat{E}^{a'}_{d} ) \hat{E}^b_j (\hat{E}^{b'}_{u}  +\hat{E}^{b'}_{d} )   |{E}^Z_k  \rangle \label{eq:line34}\\
& = \hat{E}^a_i \,\mathbb{1} \, \hat{E}^b_j \, \mathbb{1} \,  |{E}^Z_k  \rangle \label{eq:line35}\\
& = \hat{E}^a_i  \hat{E}^b_j   |{E}^Z_k  \rangle\\
& =  | {E}^a_i {E}^b_j   {E}^Z_k  \rangle\,,
\end{align} 
where lemma \ref{beamops} is used to go from line (\ref{eq:line34}) to (\ref{eq:line35}).

Now, 
\begin{align}\label{jointdcfmarginals}
\sum_{i'{\bar{i}}' j'\bar{j}' } D^{patch}\left(ii'jj'k \,,\,  \bar{i} {\bar{i}}'\bar{j}{\bar{j}}'\bar{k}\right) &=\sum_{i'{\bar{i}}' j'\bar{j}' } 
\langle  ii'jj'k  \, |\, \bar{i} {\bar{i}}'\bar{j}{\bar{j}}'\bar{k}\rangle \\
&=  \langle   
        {E}^a_i   {E}^b_j  E^Z_k  
\, | \,  {E}^a_{\bar{i}} {E}^b_{\bar{j}} E^Z_{\bar{k}}     \rangle\\
& = D^{abZ}( ijk, {\bar{i}} {\bar{j}} {\bar{k}})\,,
\end{align}
where the last line follows from the definition of event-vectors (section \ref{sec:EHS}). 
A similar calculation shows that  $D^{patch}$ has $D^{\alpha\beta Z}$ as marginals for the 3 other global settings $\alpha\beta= ab', a'b, a'b'$. 

The fact that the set of  $D^{\alpha\beta Z}$'s  is a  PoZ model (Definition \ref{def:pozmodel}) of the decoherence functionals $D^{\alpha\beta}$ (\ref{fourdcfs}) implies that further marginalization over $k$ and $\bar{k}$ gives the 4 decoherence functionals (\ref{fourdcfs}). 
\end{proof}
There is a converse of sorts: 
\begin{theorem} (Lemma 4.2 in \cite{Craig:2006ny}) If there exists a joint decoherence functional, $D_{joint}$ on all the beam-events that has the four decoherence functionals (\ref{fourdcfs}) as marginals, then one can  construct, formally, a PoZ and LoN model for (\ref{fourdcfs}) that is quantum factorizable.
\end{theorem}
\begin{proof}  We prove this by explicit construction. By assumption there exists a decoherence functional $D_{joint}\left(  ii'jj',\, \bar{i} {\bar{i}}'\bar{j}{\bar{j}}'\right) $ 
with 
(\ref{fourdcfs}) as marginals. We construct 4 measure theories, one for each global measurement, that agree in $Z$ where there are exactly 16 history-events. Each history-event in $Z$ is labelled by a 4-bit string: $\{  I I' J J' \, |\, I , I', J, J' = u, d\}$. These history-events are posited formally and the 4 $\{u,d\}$ bits are just labels and do not imply that anything in $Z$ is ``up'' or ``down''. In the decoherence functionals for the 4 measure theories,  the relevant beam-events in $A$ and $B$ are simply determined by the corresponding bit values in the past history-event in $Z$ as follows:
\begin{align}
 D^{abZ}(ij  I I' J J' \,,\, \bar{i}\bar{j}\bar{I}\bar{I}' \bar{J} \bar{ J}') &:=
  \delta_{iI}\delta_{\bar{i}\bar{I}}  \delta_{jJ}\delta_{\bar{j}\bar{J}}
  D_{joint}\left( II'JJ' ,\,\bar{I} {\bar{I}}'\bar{J}{\bar{J}}'\right) \\
 D^{ab'Z}(ij' I I' J J' ,\, \bar{i}\bar{j}'\bar{I}\bar{I}' \bar{J} \bar{ J}') &:= 
 \delta_{iI}\delta_{\bar{i}\bar{I}}  \delta_{j'J'}\delta_{\bar{j}'\bar{J}'}
  D_{joint}\left( II'JJ' ,\,\bar{I} {\bar{I}}'\bar{J}{\bar{J}}'\right)\\
 D^{a'bZ}(i'j  I I' J J' ,\, \bar{i}'\bar{j}\bar{I}\bar{I}' \bar{J} \bar{ J}') &:= 
  \delta_{i'I'}\delta_{\bar{i}'\bar{I}'}  \delta_{jJ}\delta_{\bar{j}\bar{J}}
  D_{joint}\left( II'JJ' ,\,\bar{I} {\bar{I}}'\bar{J}{\bar{J}}'\right)\\
D^{a'b'Z}(i'j'  I I' J J' ,\, \bar{i}'\bar{j}'\bar{I}\bar{I}' \bar{J} \bar{ J}') &:= 
  \delta_{i'I'}\delta_{\bar{i}'\bar{I}'}  \delta_{j'J'}\delta_{\bar{j}'\bar{J}'}
D_{joint}\left( II'JJ' ,\,\bar{I} {\bar{I}}'\bar{J}{\bar{J}}'\right)\,.
\end{align}

It can be verified that this is a PoZ model of the the four decoherence functionals (\ref{fourdcfs}) and that the
 decoherence functional for each global setting satisfies LoN and is factorizable (\ref{qfactor}). 
\end{proof}

When $i,i',j,j'$ refer to  beam-events that are detection events, the decoherence functionals  in  (\ref{fourdcfs}) will be diagonal and the diagonal elements will be the experimental probabilities (\ref{expprobs}).  The existence of a joint decoherence functional on all the beam
events then implies the Tsirel'son inequalities for the experimental probabilities \cite{Craig:2006ny} and
indeed the stronger condition known as $Q^1$ \cite{Dowker:2013tva} in
the Navascues-Pironio-Acin hierarchy
\cite{Navascues:2008}.

\medskip
In the remaining sections of this paper, we adopt a less formal tone and
explore 
the implications of 
{\rc}
in relation to 
spacelike correlations, following which we
consider further how the conditions of PoZ and LoN relate to {\rc}
and other foundational notions like locality, time-asymmetry,
and logic.
We begin by arguing that unadorned {\rc} has very little to say on the
matter of which correlations can occur in nature and which cannot.

\section{Nonlocality rescues causality!} \label{sec:rescue}
Tradition has it that the observed violation of the Bell/CHSH
inequalities implies that some sort superluminal causation is taking
place; and there are purely logical antinomies 
(e.g. Kochen-Specker-Stairs \cite{Stairs:1983}, Greenberger-Horne-Zeilinger (GHZ) \cite{Greenburger:1989}, Hardy \cite{Hardy:1992})
that are even more persuasive in this regard.  In all of these
instances the evidence for the alleged superluminality consists of
certain spacelike correlations, either probabilistic or deterministic.
Against tradition, however, we maintain that 
\textit{none of these correlations entails superluminal causation.}  
And this is true independently of
whether the context is classical or quantum.

Why then do so many people\footnote
{Representative authors are Albert Einstein, John Bell \cite{Bell:1990}, and Travis
 Norsen in his otherwise very clear and careful exposition \cite{doi:10.1119/1.3630940}.  (Similarly, some of the philosophical literature: e.g. \cite{Butterfield:1992, Butterfield:1994}.) Bell and
 Norsen do not seem to have specified the ``beables'' they had in mind;
 but Einstein (who presumably was unaware of histories-formulations)
 seems to have been taking the wave-function $\psi$ as the quantum
 model of reality, and pointing at its remote ``collapse'' as the ``spooky
 distant-action''  \cite{born1971born} he was opposing.}
feel that certain kinds of correlations 
among events in mutually spacelike regions of spacetime
{\it\/do\/} conflict with {\rc}?
We think this comes about because they are taking for granted that
causes can only act locally in spacetime, 
with the result that they slide from
{\rc} per se to
an enhanced condition like that to which John Bell gave the name `local causality'.
Let us focus instead on what {\rc} (RC) demands when taken alone,
unalloyed with any further requirement of locality or local causation.
What it wants to say is then very simple (albeit far from precisely
formalized): 
\textit{An event that happens in region $X$ cannot influence events in region $Y$ disjoint from Future($X$).}  

Now let $X$ be, for example, a spacetime region where a ``source'' emits
an entangled pair of particles.  The correlations at issue relate events
in region $A$ to events in region $B$, and both of these regions lie in the
future of the source.  The correlations themselves pertain to neither
region individually, but to their union, $Y = A \cup B$, which a 
fortiori also lies within Future($X$).  But for something that happens in
a region $X$ to cause something else to happen in that region's future in
no way conflicts with \rc.  The alleged contradiction disappears.

That is the whole argument, but perhaps 
some further comments would be helpful.
In order to feel fully at home with the above reasoning it's necessary
to grant a certain conceptual independence to events as such, as indeed
the framework of QMT does.\footnote
{A quote from Kolmogorov conveys this insight 
  [for ``elementary event'' read individual history]:
 ``the notion of an elementary event is an artificial superstructure
   imposed on the concrete notion of an event. In reality, events are
   not composed of elementary events, but elementary events originate
   in the dismemberment of composite events'' (for an English translation by R. Jeffrey of Kolmogorov's 1948 paper see \cite{Kolmogorov:1948}).}
A correlation between events in $A$ and events in $B$ is an event in its
own right, an event in {\AUB} not reducible to some event in $A$ together
with some other event in $B$.

What you must {\it\/not\/} think to
yourself is that $Z$ can cause such a correlation-event {\it\/only\/} by
separately inducing a particular $A$-event and a particular $B$-event.
If we are right, it is the unacknowledged (or only partly acknowledged)
embrace of this intuition of locality that creates the apparent conflict
between quantum mechanics and relativistic causality.  It thus seems
worthwhile to try to identify the minimum hypothesis (weaker than Bell's
local causality) that needs to be given up if one is to avoid the conflict.
We will state this hypothesis (or principle) in relation to our two
spacelike-separated regions, $A$ and $B$, although the formulation can
be extended in an obvious manner to any set of disjoint spacetime-regions.

\begin{definition} [Principle of Sheafy Causation] 
It asserts the following.  A cause can influence an
event in region $A\cup B$ only by {\it\/influencing separately\/}
events in $A$ and events in $B$:
a cause's effect/action in $A \cup B$ is {\it\/fully given\/} by its
effect/action in $A$ together with its effect/action in $B$.
\end{definition}
\noindent 
Adopting a turn of phrase popular in the philosophical literature, one
might say that according to this principle, a cause's action in {\AUB}
``{supervenes on}'' its action in $A$ and its action in $B$.
(We have chosen the adjective ``sheafy'' because a mathematical sheaf is the
paradigmatic object whose global properties supervene on its local
properties.)

\subsection{Dispelling the supposed paradox}

We can illustrate how sheafy causation leads to a paradox by adopting it
provisionally, and then reasoning about a minimal fragment of the EPRB
Gedankenexperiment, the fragment concerning perfect correlations.

Let there be two Stern-Gerlach analyzers with fixed settings so aligned
as to produce a perfect correlation between the respective beams ($u$ or
$d$) in which the $A$- and $B$-particles emerge.  Following a suitable
preparation-event in region $X\subseteq Z$, the $A$-particle emerges from its
analyzer in the upper beam if and only if the $B$-particle does the same.
In other words, the event ``$uu$ or $dd$'' (which we will henceforth
write as $uu\!+\!dd$ or $uu\!\cup\!dd$) must happen.

Now causality insists that particle-$A$ cannot learn what particle-$B$
is doing, and particle-$B$ cannot learn what particle-$A$ is doing.
Therefore both particles must know in advance whether they will choose
their `$u$' beam or their `$d$' beam.  Hence the source or preparation
must have {\it\/pre-determined\/} both of these separate ``choices''.
(Or so it seems!)

It is this supposed predetermination of the outcomes, $uu$, $dd$, $ud$,
or $du$, that is the source of the paradoxes (basically because it
authorizes the introduction of ``hidden variables'' $\lambda$ that bear
the information about which outcome has been determined to happen in any
particular run.)  To undermine the reasoning that seemed to lead to
deterministic beam-events is thus to dispel the paradoxes.

What, then, was wrong with the reasoning we just rehearsed?  
The fallacy, as we have already indicated, 
was to have conflated causation per se with sheafy causation:
to have ignored that 
events 
in $X$ can cause {\it\/any\/} event in $X$'s future,
without necessarily being the cause of any other event.
In particular they can cause $uu+dd$ to happen, without needing to cause
either $uu$ or $dd$.
That is, 
a cause
need not (and in this example does not)
influence the particles individually, but only jointly.

If we accepted the principle of sheafy causation, we would have to deny
this possibility. Instead, we would infer that in order to force
$uu+dd$ to happen a cause would either have to force $u$-left and
$u$-right to happen (and thus force $uu$ to happen) or else force $dd$
to happen.  In some experimental runs, the fine details of the
preparation-event would deterministically produce the $u$-event in each
wing, while in other runs they would produce $d$ in each wing.

It might be helpful to express our view of 
the situation
in negative language. The
preparation-event has {\it\/prevented\/} the events $ud$ and $du$ from
happening.  Beyond that it has done nothing, having had in particular
zero influence on $A$ and zero influence on $B$.  It has caused a
correlation, but no more than that.\footnote%
{One might challenge the words, ``zero influence'', by claiming that the
  preparation ``caused events $u$-left and $d$-left in $A$ to be equally
  likely'', and similarly for region $B$.  However, even if one accepts
  this as a causal influence, it does nothing toward producing the
  \textit{correlation} between $u$-left and $u$-right.}

\medskip
\noindent
{\bf Note}. An event like $uu+dd$ is a logical (Boolean) combination of events
in $A$ and events in $B$: it ``supervenes on'' these events.  
This is indeed a species of locality, but it is
purely kinematical, while the crucial nonlocality
is dynamical. The {\it\/causal influence\/} which the $X$-event exerts
in region {\AUB} does {\it\/not\/} supervene on its influences in
regions $A$ and $B$ separately.

Perhaps we should also clarify here that by focussing on the correlation
event $uu+dd$ we do not mean to endorse an assertion like ``Neither $uu$
nor $dd$ happens in nature, but only $uu+dd$''.  Such a claim would not
make sense, given that both $uu$ and $dd$ are macroscopic events.
The present paper revolves around questions of cause, locality, and
logic.  We do not address the measurement-problem, which we would view
as the task of explaining why macro-reality can be identified with a
single (macro-)history, even 
if
micro-reality cannot.  Empirically
however, this is a fact, which can also be expressed by saying that
macro-events follow classical logic.  To explain this fact is a task for
``Quantum Foundations'', but if we take it as given, then it follows at
the macroscopic level that $uu+dd$ happens $\implies$ either $uu$
happens or $dd$ happens.  However, this still doesn't imply that the
preparation-event \textit{causes} $uu$ or \textit{\/causes\/} $dd$.  What it
causes is still just $uu+dd$.\footnote%
{In observing that reality is described by a single history at the
  macroscopic level, we are not claiming the same about microscopic
  reality.  That the course of microscopic events involving a given
  particle could correspond to a single worldline of that particle would
  contradict \RC as we understand it, as illustrated by the purely
  logical cousins of the EPR paradox.
 Nor do we mean to imply that the particle detectors in the $u$ or $d$
 beams ``only reveal'' the locations of the particles they are
 detecting. Nor do we mean to imply that they don't!}

\subsection{How does the path-integral explain the perfect correlations?}

If some of these these explanations seem unduly abstract or slippery, it
might help to go through the path-integral calculation presented in
\cite{Sinha:1991cj} that shows in detail how the correlations come
about.  One sees in particular how the precluded event $ud$ acquires a
net amplitude of zero.
As one will readily observe, the calculation is {global in nature},
because the amplitudes that enter into it are themselves global in
nature. (They are functions of histories.)  
Within a path-integral framework, the fact that a correlation is an
essentially nonlocal effect is clearly visible in the computation that
one performs to deduce the correlation.  One sees concretely how 
the preparation-event exerts its causal influence
globally without doing so locally.  Event $u$ at $A$ acquires
a positive quantum measure $\mu(u)>0$, as does event $d$ at $B$;  but the
measure of the intersection $ud$  (their conjunction) vanishes.

In order to avoid a possible confusion, we should mention here that the
setup in \cite{Sinha:1991cj} differs from that discussed in
Section \ref{sec:patch} in one important respect.  The beam-events
considered in Section \ref{sec:patch} were macroscopic
instrument-events: the registering of the presence of one or more
particles by one or more detectors placed in the corresponding beams.
In contrast, the computation in \cite{Sinha:1991cj} did not include
detectors, and indeed probably no one has attempted something like that
with realistic detector-models.\footnote%
{It's interesting that realistic source-models  are
   much easier to devise.  In the arrangement of \cite{Sinha:1991cj}
   a single mirror (or beam-splitter) suffices to ``entangle'' 
   a pair of photons with each other.  The design takes
   advantage of the photons' bosonic statistics, and it requires no
   nonlinearity of the kind involved in parametric down-conversion.}
Rather, the events whose measures $\mu$
were computed in \cite{Sinha:1991cj} were the corresponding beam-events
without detectors present.

This simplification is always made in practice when people compute
observational probabilities, but of course it ultimately needs to be
justified --- something which will only be achieved fully when the
so-called measurement problem has been solved.  Pending that, we can
perhaps be content with: (1) the assumption (or widespread conviction)
that if the computation could be done, 
the measure $\mu$ of the particles emerging in certain beams without
detectors present, would equal the measure of detectors placed in the
same beams registering the particles' presence;
together with 
(2) the rule of thumb that the measure of an instrument-event can be
interpreted as a probability in the sense of a relative frequency.
Together these are equivalent to the Born Rule.

\subsection{For nonlocality}
Not too many years ago, drawing a distinction between ``local
causality'' and causality per se might have seemed to be splitting
hairs, but now physicists possess many reasons to take a fundamental
nonlocality seriously; and most of these reasons have nothing to do with
the Bell inequalities.  One can mention here the puzzle of the
cosmological constant $\Lambda$; the continuing interest in nonlocal
field theories, non-commutative geometry, and twistors; and the fact
that (as illustrated by causal sets) a spatio-temporal discreteness can
be combined with Lorentz invariance only by accepting a radical
nonlocality.  Indeed, if spacetime is ultimately discrete, then locality
will largely lose its meaning simply because the concept of
infinitesimal neighborhood of a point will no longer be available.
But even a relatively limited amount of nonlocality would erase any
difference of principle between causing an event in a small neighborhood
of a spacetime point and causing an event in a much larger region,
and this in turn would suffice to undermine the ``sheafy'' reduction of
a causal influence on the amalgamated region $A+B$ to separate
influences on the constituent regions $A$ and $B$.

\subsection {Correlations involving multiple instrument-settings}
The discussion above pertains to analyzers whose settings are fixed,
whereas (as with patching) the Bell inequalities and the
gedankenexperiments relating to superluminality all require the
consideration of multiple settings.  
However, 
the case of variable settings is more general in appearance only. It 
can be handled in the same way as the fixed case 
if one treats the settings as
the dynamical events they actually are (so enlarging the history-space and event-algebra to
include the instruments and their histories).
In place of a
correlation event like `$uu+dd$',
one now puts an event 
like \hbox{`$S\to(uu+dd)$'}, 
where $S$ denotes a setting-event,
$u$ and $d$ denote beam-events when the Stern-Gerlachs have been oriented by $S$,
and $\to$ denotes the Boolean operation of so-called
``material implication'', defined by 
\begin{equation}
   E \to{F} := (\Omega\backslash E) \cup F.
\end{equation}
Instead of saying that the preparation event in region $X\subseteq Z$
causes in region $\AUB$ the event, $uu+dd$, to happen, one now says that
it causes the event, $S\to(uu+dd)$, to happen.\footnote%
{Made explicit as a set of histories, this event is $(\Omega\backslash S)\cup uu \cup dd$.
 In words these are the histories such that either the correlation event happens or the settings are not those of $S$.}
In this manner, perfect correlations involving multiple instrument-settings can be handled exactly as above.
In particular this covers the case of the EPRB setup with variable, but matched, settings of the analyzers.

For a more generic example, consider the correlations of the
Popescu-Rohrlich-boxes \cite{Popescu:1994}, which can be expressed as
follows.
Let $E_{PR}$ be the event\footnote%
{In words: If the joint setting is 1-1, 1-2, or 2-1 then the beam-events
 are perfectly correlated, and if the joint setting is 2-2 then the
 beam-events are perfectly anti-correlated.}
\be
  \big[(S_1 S_1 \cup S_1 S_2 \cup S_2 S_1) \; \to \; (uu \cup dd)\big] \; \cap\; \big[S_2 S_2 \,\to\, (ud \cup du)\big]
\ee
where $S_1$ and $S_2$ are settings, 
$u$ and $d$ are beam-events as before, 
and where for instance
$S_2 S_2$ denotes the event ``setting $S_2$ at both $A$ and $B$''.
The causal influence in this case can be expressed by saying that the
preparation event $P$ causes event $E_{PR}$ to happen.
[Alternatively, one could say that it causes four different events to happen, namely the events,
$S_1 S_1\to (uu \cup dd)$,  $S_1 S_2\to (uu \cup dd)$, $S_2 S_1\to (uu \cup dd)$, and $S_2 S_2\to (ud \cup du)$.]

As a final example (one which is more amenable to experiment), 
consider the 3-beam GHZ correlations \cite{Greenburger:1989}, 
which are encapsulated in the event, 
\begin{align*}
 E_{GHZ}& = \big[ \left( S_x S_y S_y  \cup S_y S_x S_y \cup S_y S_y S_x \right) \; \to \; \left( uud   \cup udu \cup duu\cup ddd \right)
 \big] \\
 & \quad\quad \quad  \cap\; \big[  S_x S_x S_x  \; \to \; \left( ddu   \cup dud \cup udd \cup uuu \right)
 \big] 
\end{align*} 
where $S_x$ and $S_y$ are again settings.
In this case, the causal influence would be expressed by saying that the preparation event $P$ caused $E_{GHZ}$ to happen.

In these examples one is still dealing with perfect correlations, but
there's no need to go further if one's interest is in the conflict
between {\rc} and locality.  Indeed the perfect GHZ correlations, for
instance, are more trenchant in that respect than the merely
probabilistic correlations involved in the CHSH/Bell inequalities.
Nevertheless, the Bell inequalities are still the most relevant to
accomplished experiments, and one can ask how the discussion of this
section would look in relation to them.
More generally, how should {\rc} be conceived in the context of
probabilistic correlations?

In that context,
one is dealing with a broader and less transparent concept that one might
term ``stochastic causation'', and it seems clear that cause-effect
implications of a straightforward logical nature no longer suffice.
Instead of statements like ``$P$ causes $Q$'', it seems that one would need
to make sense of locutions like ``$P$ causes $Q$ with probability $p$'',
or (more obscurely) ``$P$ causes $Q$ to have the probability $p$'', or
perhaps even ``a string of repetitions of $P$ causes the frequency of
outcome $Q$ to be $p$.''
But these matters concern stochastic causation and probability as such,
and are not really germane in the present context.
(Recall here that because the setting- and beam-events under discussion are
macroscopic instrument-events, quantal interference is absent, whence one can employ
ordinary (``homomorphic'') logic and ordinary probability theory in
reasoning about them.)

\section{Why are certain correlations \textit{not} seen?} \label{sec:separability}
The message from the preceding section is that
the principle of relativistic causality imposes no
limitation on the correlations that a localized cause can induce among
events in its future.  In particular, \rc {} is perfectly
compatible with the spacelike correlations that have often been supposed
to contradict it 
(up to and including hypothetical ``signalling correlations'').
Conflicts arise when you supplement {\rc} with sheafy causation,
or with some still stronger condition like continuous propagation of
cause-effect chains in spacetime (Bell's Local Causality). 
But experience teaches that these stronger principles are frequently
violated in the quantal world, and therefore must be given up.

If this were the end of the story, however, then one might expect to have
observed in nature correlations much stronger than those discovered so far, 
and in particular stronger than provided for in established quantum theories,
which respect constraints like ``no-signalling'' and the Tsirel'son
inequalities (the simplest of the latter being CHSH with $2\to 2\sqrt{2}$ \cite{Tsirelson:1980,Tsirelson:1985}.)
This suggests that some other principle beyond bare {\rc} is active in
nature, and that it might be encoded in certain structural features
of the quantum measure which lie at the base of more phenomenal
regularities like the ``patching property'' and its consequences for
correlations, like the ``no-signalling'' equalities.
If the explanation for these regularities is indeed some hitherto
unrecognized structural principle governing the decoherence functional, then
discovering what it is would not only help to illuminate our current theories,
but it might usefully guide the search for new theories, especially theories
of quantum gravity.
Without really knowing how to frame such a principle, we will as a
placeholder give it the name of {\it\/\caussep\/},
and try to indicate how PoZ
might be a step in its direction.

In ordinary quantum mechanics, one  proves the Tsirel'son
inequalities by assuming that the correlators they relate can be expressed
as expectation values of products of projection operators.  In quantum
measure theory, 
as already mentioned, 
the same inequalities 
follow
from the hypothesis of a joint quantum measure, 
the analog of a joint probability measure in Fine's theorems.  
That is,
they can be derived from ``quantal patching''.  
A would-be principle of \/\caussep\/ would thus want to provide a basis for quantal patching.

In the  ``decoherent histories'' and ``consistent histories'' 
interpretations of quantum mechanics \cite{Griffiths:1984rx, Omnes:1988ek,Gell-Mann:1989nu, Hartle:1989})
the decoherence functional is commonly \textit{defined} in
terms of sequences of projection operators in a fixed Hilbert space.
If we could assume that this were its most general form,
then patching would follow rather simply,
but that assumption is not
tenable in the context of path-integrals, let alone in quantum gravity.
Fortunately it is not needed either, because one can appeal 
instead to event-operators, as 
we have seen
in Section \ref{sec:patch} of this paper.

But this in turn makes us ask what kind of condition would ensure that 
the required event-operators will be available?  
As detailed above, 
one possible answer, 
is that such a condition is PoZ, 
or rather PoZ supplemented by what we have termed Lack of Novelty (LoN).
In light of this service that PoZ provides to quantal patching, 
one can see
it as a species of causality principle, 
one that 
(thanks to the ``slicing freedom'' inherent in a Lorentzian temporal structure)
is able to play a similar role quantum mechanically 
to what Bell's Local Causality played in the derivation of the CHSH inequalities,
or to what the Principle of Sheafy Causation plays in the derivation of
deterministic hidden variables from the perfect correlations of the original EPR paper.
(In relation to CHSH one has schematically
that,
classically: 
  Local Causality $\to$ Factorizability $\to$ classical patching $\to$ CHSH;
and quantally:
  PoZ + LoN $\to$ event-operators $\to$  quantal patching $\to$ Tsirel'son.)
Moreover, quantal patching also ensures, essentially by definition,
that the resulting correlations will satisfy the condition on the
marginal probabilities that goes by the name of ``no-signalling''.
This, then, is one  way in which PoZ relates to ``\/\caussep\/''.  

But PoZ has other features, too, that make contact with our causal
intuitions.  First and foremost 
(and in contrast to other ``causality principles'' like spacelike commutativity) 
it 
manifests
the essential
time-asymmetry inherent in the causality-concept: the time-reverse of
PoZ is a condition that will almost never be satisfied!

Moreover, the particular manner in which PoZ provides this ``arrow of time'',
relates to irreversibility and the ``stability of the past'',
because
it can be read as saying that a certain kind of ``property of the
past'' cannot be undone in the future.
The preclusion of a past event is 
such a property,
and although the full statement of PoZ goes beyond simple event-vectors to
linear combinations of them, one can perhaps regard the vanishing of
a sum like that on the LHS of (\ref{eq:pozero}) as also being a
property of the past.

Finally, we should comment on the possibility that certain
causality-principles could retain their heuristic value for theories in
which the metric, and therefore the temporal/causal structure of
spacetime, becomes dynamical, in other words for theories of quantum
gravity.  In that treacherous soil, it seems very possible that none of
our cherished causality principles will be able to take root.
However, it's worth mentioning that for decoherence functionals defined
on the event-algebra of labeled causal sets there does exist a natural
analog of the PoZ condition.

\section{Acknowledgments}

For stimulating questions touching the topic of this paper,
   RDS would like to thank participants in the ``Tonyfest"
   symposium, held 3 February 2019, at the Raman Research
   Institute in Bengaluru.  
 FD  acknowledges the support of 
 the Leverhulme/Royal Society interdisciplinary  APEX grant APX/R1/180098.  
 FD is supported in part by STFC grant ST/P000762/1. 
Research at Perimeter Institute is supported by the Government of Canada through
Industry Canada and by the Province of Ontario through the Ministry of Economic Development and Innovation. 
RDS is supported in part by NSERC through grant RGPIN-418709-2012.

\appendix
\section{Appendix}\label{appendix}

Consider  a unitary local quantum field theory on globally hyperbolic
spacetime $M$. Let the field be $\Phi$ and let $\Sigma_i$ be some Cauchy
surface in the past on which there is an initial state $\psi$ which is a
set of amplitudes $\psi[\varphi]$ for each spatial configuration
$\varphi = \Phi|_{\Sigma_i}$ on $\Sigma_i$.  
The decoherence functional for events $E$ and $\bar{E}$ in a spacetime
region $R$ between the initial Cauchy surface $\Sigma_i$ and a final
Cauchy surface $\Sigma_f$ is given by a double path integral of
Schwinger-Keldysh type: 
\begin{align} \label{skform} 
D(E, \bar{E}) = \int_{\Phi\in E, \bar{\Phi} \in \bar{E}}\, {\cal{D}}\Phi {\cal{D}}\bar{\Phi} \,\,
\psi\left[\Phi |_{\Sigma_i}\right]^* \,  \psi\left[\bar{\Phi} |_{\Sigma_i}\right] \,\,
e^{ -i S\left[\Phi \right] + i S\left[\bar{\Phi}\right]} \,\, \delta\left[ \Phi |_{\Sigma_f} -  \bar{\Phi} |_{\Sigma_f} \right] \,,
 \end{align}
where the delta functional enforces the condition that the two histories $\Phi$ and $\bar\Phi$ are equal on the final 
truncation surface
$\Sigma_f$. This path integral may also be thought of as a single
integral over what we call \textit{Schwinger histories}
\cite{doi:10.1063/1.1337730} consisting of the pair $(\Phi, \bar{\Phi})$
that agree on $\Sigma_f$. 

It is a property of unitary theories that the value of the decoherence
functional does not depend on the position of the truncation surface
$\Sigma_f$, so long as  $\Sigma_f$ is nowhere to the past of events $E$
and $\bar{E}$.  

Now consider regions $Z$, $A$ and $B$ as in the article, as shown in figure \ref{frog}. 
We denote the union of $Z$ and 
$A$ by $Z+A$, the union of $Z$ and 
$B$ by $Z+B$ and the union of $Z$, $A$  and 
$B$ by $Z+A+B$, as in the article. 
We denote the future boundary of a region  $X$ by $\partial^+ X$ and its past boundary by $\partial^- X$.

\begin{theorem} Consider the unitary quantum field theory with
  decoherence functional as described above. Let  regions $Z$, $A$ and
  $B$ be as in figure \ref{frog}. Let $E_A$ and $\bar{E}_A$ be events in
  $A$, $E_B$ and $\bar{E}_B$ be events in $B$ and let  $E_{\gamma_Z}$ and
  $E_{\bar{\gamma}_Z}$ be history-events in $Z$. Then    
\begin{align}\label{qfactorize}
D(E_A E_B  E_{\gamma_Z},\bar{E}_A \bar{E}_B  E_{\bar{\gamma}_Z}) \, D(E_{\gamma_Z}, E_{\bar{\gamma}_Z}) = D(E_A  E_{\gamma_Z}, \bar{E}_A  E_{\bar{\gamma}_Z})\,  D(E_B  E_{\gamma_Z}, \bar{E}_B  E_{\bar{\gamma}_Z}) \,.
 \end{align}
\end{theorem}
\begin{proof}
Recall that a history-event  in $Z$ is  the cylinder set defined by a
single field configuration on $Z$ and let us refer to the field
configurations on $Z$ corresponding to history-events $E_{\gamma_Z}$ and
$E_{\bar{\gamma}_Z}$ as  
$\gamma$ and $\bar{\gamma}$ respectively. 

The portion of the initial surface $\Sigma_i$ that is not contained in
$Z$ can be ignored and, instead of initial and final Cauchy surfaces, we
have initial and final partial Cauchy surfaces: the initial surface is
$d := \partial^- Z \cap \Sigma_i$  and the final, truncation surface is
the future boundary of the relevant region, $Z$, $Z+A$, $Z+B$ or $Z+A+B$
for the decoherence functional in hand. See figure \ref{frog4}.  The
initial amplitudes $\psi$ are defined for spatial field configurations
on  the initial partial Cauchy surface, $d$.  

We consider the four decoherence functionals in the identity
(\ref{qfactorize}) in turn starting with the simplest, $ D(E_{\gamma_Z},
E_{\bar{\gamma}_Z})$. For this, there is no path integral at all because we
can choose the truncation surface to be the future boundary of $Z$: 
\begin{align}
D(E_{\gamma_Z}, E_{\bar{\gamma}_Z}) =  
           \psi\left[\gamma |_{d}\right]^* \,  \psi\left[\bar{\gamma}
             |_{d}\right] \,\, e^{ -i S\left[\gamma\right] + i
             S\left[\bar{\gamma}\right] } \,\,\delta\left[  \gamma
             |_{\partial^+Z} - \bar{\gamma} |_{\partial^+Z}\right]\,.        
 \end{align}

$\partial^+ Z$ can be partitioned into three: $a: = \partial^- A \cap \partial^+ Z $, 
$b: = \partial^- B \cap \partial^+ Z $ and the complement $c : = \partial^+ Z \setminus (a \cup b)$
as shown in fig \ref{frog4}. 
 \begin{figure}[h!]
\centering
{\includegraphics[scale=0.25]{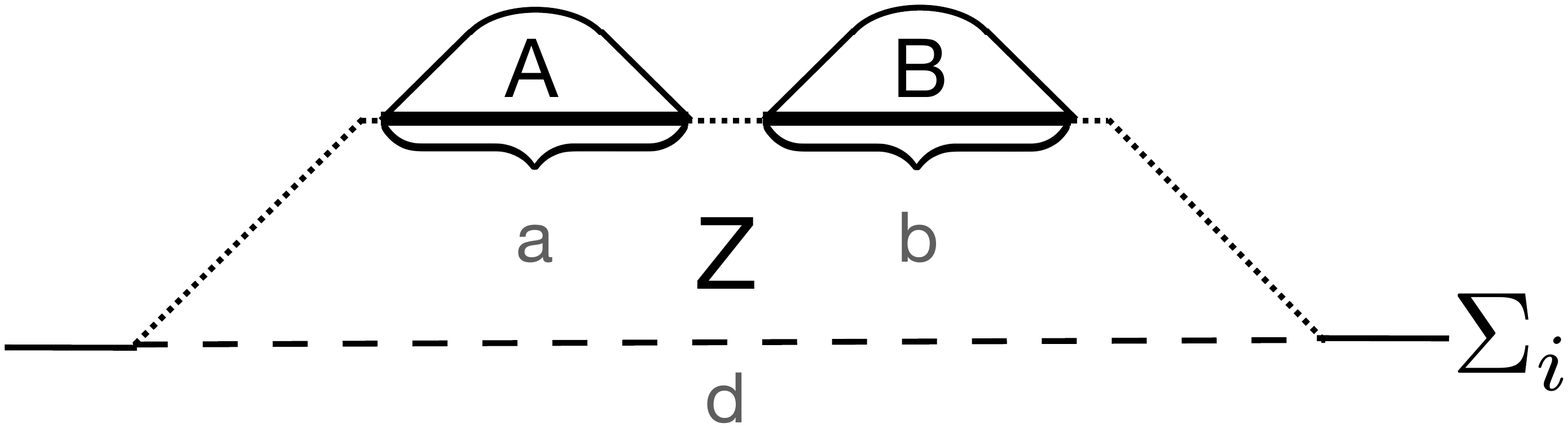}}
\caption{\label{frog4} The horizontal line $\Sigma_i$ is the initial Cauchy surface. The dashed portion of $\Sigma_i$ represents $d:= \partial^-Z\cap \Sigma_i$. The future boundary of $Z$ can be partitioned into 3 sets: 
$a: = \partial^- A \cap \partial^+ Z $, 
$b: = \partial^- B \cap \partial^+ Z $ and the complement $c : = \partial^+ Z \setminus (a \cup b)$.
$a$ and $b$ are shown as bold lines and $c$ is indicated by the dotted line.  }
\end{figure}
The delta function for the fields on $\partial^+ Z$ is  therefore a product of  3 delta functions for the field on $a$, $b$ and $c$: 
\begin{align}
\delta\left[ \gamma |_{\partial^+ Z} -  \bar{\gamma} |_{\partial^+ Z} \right]
= \delta\left[ \gamma |_{a} -  \bar{\gamma} |_{a} \right]\,\, \delta\left[ \gamma |_{b} -  \bar{\gamma} |_{b} \right]\,\, \delta\left[ \gamma |_{c} -  \bar{\gamma} |_{c} \right] \,.      
 \end{align}

Now consider

\begin{align}
&D(E_A E_{\gamma_Z}, \bar{E}_A E_{\bar{\gamma}_Z})  \\
& = \int_{\substack{\Phi\in E_A, \bar{\Phi} \in \bar{E}_A \\ \Phi |_Z = \gamma, \bar{\Phi} |_Z = \bar{\gamma}}} {\cal{D}}\Phi {\cal{D}}\bar{\Phi} \,
 \psi\left[\gamma |_{d}\right]^* \,  \psi\left[\bar{\gamma} |_{d}\right] 
 \,\, e^{ -i S\left[\Phi\right] + i S\left[\bar{\Phi}\right] } \,
\delta\left[  \Phi  |_{\partial^+(Z+A)} - \bar{\Phi} |_{\partial^+(Z+A)}\right]\\     
 & =    \psi\left[\gamma |_{d}\right]^* \,  \psi\left[\bar{\gamma} |_{d}\right] 
 \,\, e^{ -i S\left[\gamma\right] + i S\left[\bar{\gamma}\right] }  \delta\left[  \gamma  |_{c} - \bar{\gamma} |_{c}\right] \delta\left[  \gamma  |_{b} - \bar{\gamma} |_{b}\right]\\ & \quad \quad \quad \quad \quad \quad \quad \int_{\Phi\in E_A, \bar{\Phi} \in \bar{E}_A }
{\cal{D}}\Phi {\cal{D}}\bar{\Phi} \,\, e^{ -i S\left[\Phi|_A\right] + i S\left[\bar{\Phi}|_A\right]} \delta\left[ \Phi  |_{\partial^+A} - \bar{\Phi} |_{\partial^+A}\right]  \,,
 \end{align}
where the path integrals in the last line are over field configurations on region $A$ only.

Similarly, 
\begin{align}
&D(E_B E_{\gamma_Z}, \bar{E}_B E_{\bar{\gamma}_Z})  \\
 & =    \psi\left[\gamma |_{d}\right]^* \,  \psi\left[\bar{\gamma} |_{d}\right] 
 \,\, e^{ -i S\left[\gamma\right] + i S\left[\bar{\gamma}\right] }  \delta\left[  \gamma  |_{c} - \bar{\gamma} |_{c}\right] \delta\left[  \gamma  |_{a} - \bar{\gamma} |_{a}\right]\\ & \quad \quad \quad \quad \quad \quad \quad \int_{\Phi\in E_B, \bar{\Phi} \in \bar{E}_B }
{\cal{D}}\Phi {\cal{D}}\bar{\Phi} \,\, e^{ -i S\left[\Phi|_B\right] + i S\left[\bar{\Phi}|_B\right]} \delta\left[ \Phi  |_{\partial^+B} - \bar{\Phi} |_{\partial^+B}\right] \,,
 \end{align}
and the path integrals are over field configurations on region $B$ only. 

Finally, 
\begin{align} \label{ABZ}
&D(E_A E_B E_{\gamma_Z}, \bar{E}_A\bar{E}_B E_{\bar{\gamma}_Z}) \nonumber \\  
& =  \psi\left[\gamma |_{d}\right]^* \,  \psi\left[\bar{\gamma} |_{d}\right] 
 \,\, e^{ -i S\left[\gamma\right] + i S\left[\bar{\gamma}\right] }  \delta\left[  \gamma  |_{c} - \bar{\gamma} |_{c}\right] \nonumber\\ 
 & \quad \quad \int_{\Phi\in E_AE_B, \bar{\Phi} \in \bar{E}_A\bar{E}_B }
{\cal{D}}\Phi {\cal{D}}\bar{\Phi} \,\,  e^{ -i S\left[\Phi|_A\right] + i S\left[\bar{\Phi}|_A\right]} e^{ -i S\left[\Phi|_B\right] + i S\left[\bar{\Phi}|_B\right]} 
\\
& \quad \quad \quad  \delta\left[  \Phi  |_{\partial^+A} - \bar{\Phi} |_{\partial^+A}\right]  \delta\left[  \Phi |_{\partial^+B} - \bar{\Phi} |_{\partial^+B}\right] \,.  \nonumber
\end{align}
The regions $A$ and $B$ are disjoint, the integrand is a product of  factors that depend only on the field in $A$ and factors that depend only on the field in $B$. So the double path integral in \ref{ABZ} factorizes into a double path integral over fields on $A$ and a double path integral over fields on $B$:
\begin{align}
&D(E_A E_B E_{\gamma_Z}, \bar{E}_A\bar{E}_B E_{\bar{\gamma}_Z}) \\  
& =  \psi\left[\gamma |_{d}\right]^* \,  \psi\left[\bar{\gamma} |_{d}\right] 
 \,\, e^{ -i S\left[\gamma\right] + i S\left[\bar{\gamma}\right] }  \delta\left[  \gamma  |_{c} - \bar{\gamma} |_{c}\right]\\ & \quad \quad \int_{\Phi\in E_A, \bar{\Phi} \in \bar{E}_A }
{\cal{D}}\Phi {\cal{D}}\bar{\Phi} \,\, 
e^{ -i S\left[\Phi|_A\right] + i S\left[\bar{\Phi}|_A\right]} \,\, \delta\left[  \Phi  |_{\partial^+A} - \bar{\Phi} |_{\partial^+A}\right] \\
& \quad \quad \int_{\Phi\in E_B, \bar{\Phi} \in \bar{E}_B }
{\cal{D}}\Phi {\cal{D}}\bar{\Phi} \,\, e^{ -i S\left[\Phi|_B\right] + i S\left[\bar{\Phi}|_B\right]} 
\,\, \delta\left[  \Phi  |_{\partial^+B} - \bar{\Phi} |_{\partial^+B}\right]\,.
\end{align}

Putting this together, we see that the $\psi$ factors on the LHS and on the RHS of (\ref{qfactorize}) are equal,
as are the $e^{ -i S\left[\gamma\right] + i S\left[\bar{\gamma}\right] }$ factors,  as are the delta function factors for the $\gamma$ and $\bar{\gamma}$ histories on the future boundary of $Z$. We also see that the double path integral factors are equal on the LHS and on the RHS of (\ref{qfactorize}). Hence the result.

\end{proof}

\bibliography{../../Bibliography/refs}
\bibliographystyle{../../Bibliography/JHEP}


\end{document}